\documentclass[aps,pra,reprint,superscriptaddress]{revtex4-2}

\usepackage{graphicx}

\usepackage{amsmath}
\usepackage{amsthm}
\usepackage{amsfonts}
\usepackage{amssymb}
\usepackage{mathtools}
\usepackage{mathrsfs}
\usepackage{latexsym}
\usepackage{bm}
\usepackage{physics}
\usepackage{tikz}
\usepackage{here}
\usepackage[subrefformat=parens]{subcaption}

\theoremstyle{plain}
\newtheorem{thm}{Theorem}
\newtheorem{cor}[thm]{Corollary}

\newtheorem{prop}[thm]{Proposition}

\theoremstyle{definition}
\newtheorem{dfn}{Definition}

\theoremstyle{remark}

\allowdisplaybreaks
\begin{document}

\title{On the operational and algebraic quantum correlations}

\author{Shun Umekawa}
 \email{umeshun2003@g.ecc.u-tokyo.ac.jp}
 \affiliation{Department of Physics, The University of Tokyo, 5-1-5 Kashiwanoha, Kashiwa, Chiba 277-8574, Japan}
 
\author{Jaeha Lee}
 \email{lee@iis.u-tokyo.ac.jp}
 \affiliation{Institute of Industrial Science, The University of Tokyo, 5-1-5 Kashiwanoha, Kashiwa, Chiba 277-8574, Japan}
 

\begin{abstract}
We investigate the intrinsic ambiguity in the definition of correlation functions arising from the inevitable invasiveness of quantum measurements.
While algebraic correlations defined as expectation values of products of observables are widely used, their relationship to operational ones defined through actual measurement procedures remain unclear.
We demonstrate that the differences among various definitions of correlation functions and those among their underlying (quasi-)joint probability distributions are bounded above by a quantitative measure of measurement invasiveness.
We further obtain a lower bound on the discrepancy among operational and algebraic (quasi-)joint probability distributions, providing a new form of the uncertainty relation.
In addition, we identify an equivalence condition under which operational and algebraic correlations coincide.
As an application, we analyze the quantum violation of the Leggett-Garg inequality and clarify the structural origin of the equivalence among different approaches to observing the violation, including sequential projective measurements and weak-measurement.
Our results provide an operational foundation for the commonly used algebraic concepts of quantum theory.
\end{abstract}

\maketitle

\section{Introduction}
\label{Introduction}
Correlation functions play a significant role in many areas of physics.
However, in quantum theory, correlation functions involve arbitrariness in their definitions due to the incompatibility of quantum measurements.
As is widely known as the uncertainty relation \cite{Heisenberg1927,Kennard1927,Robertson1929,Maassen1988,Arthurs1988,Ozawa2003,Busch2007,Watanabe2011,Erhart2012,Lee2024}, in general, two quantum observables cannot be measured simultaneously with arbitrary precision, and this prevents the convincing definition of the correlation function.

Roughly speaking, definitions of correlation functions can be divided into two categories.
One is the \textit{operational correlation}, which is defined through the actual measurement procedure, and the other is the \textit{algebraic correlation}, which is defined as the expectation value of an algebraic ``product" of quantum observables.
A prominent example is the temporal correlation between observables \(A\) and \(B\) at times \(t_1\) and \(t_2\), respectively.
While the expectation value of the time-ordered product \(\langle B(t_2)A(t_1)\rangle_\rho\) of observables \(A(t_1)\) and \(B(t_2)\) in the Heisenberg picture is often employed as a definition of temporal correlation, it generally differs from the correlation obtained via sequential projective measurements of \(A\) at time \(t_1\) followed by \(B\) at \(t_2\).

The distinction among the definitions of quantum correlations was discussed in \cite{Fujikura2016, Shimizu2017} in the context of linear response theory.
They pointed out that the so-called ``canonical time correlation" in a quantum extension of the fluctuation-dissipation theorem differs from the temporal correlation obtained by the actual measurement process.
Then they argued that the algebraic correlation given by the symmetric product is obtained by what they call the ``quasiclassical measurement" in the thermodynamic limit.

In this work, we show that the difference among operational and algebraic quantum correlations are bounded above by a measure of \textit{invasiveness} of a quantum measurement.
Here, the term ``invasiveness" \cite{Kofler2013,Schild2015,Moreira2019} refers to the extent to which the initial state is substantially disturbed by the measurement, beyond the mere update of observer's knowledge.

We also investigate the difference among underlying (quasi-)joint probability distributions that generate operational and algebraic quantum correlations.
Quasi-joint probability distributions \cite{Wigner1932,Kirkwood1933,Dirac1945,Margenau1961,Cohen1966,Lee2017,Lee2018} are quantum analogs of the genuine joint probability distribution for a tuple of incompatible quantum observables and have various applications in quantum foundations \cite{Kenfack2004,Spekkens2008,Hofmann2015,Halliwell2016,Bievre2021,Nogami2025} and quantum computations \cite{Veitch2012,Mari2012,Raussendorf2017}.
Contrary to operational correlations that are defined through actual measurement procedures, algebraic correlations do not have an underlying joint probability distribution. 
Nevertheless, within the framework of quantum/quasi-classical representations \cite{Lee2017,Lee2018}, we admit quasi-joint-probability distributions as such distributions that naturally reproduce algebraic correlations.

The discrepancy between the operational probability and a quasi-probability was discussed in, for example, \cite{Johansen2007}.
Here, we present a quantitative evaluation of their differences and provide both upper and lower bounds among them.
While the upper bounds ensure that the differences among several definitions remain limited and are a consequence of the invasiveness and incompatibility of quantum measurements, the lower bounds present a new form of the uncertainty relations.

As an application of our discussion, we analyze the quantum violation of the Leggett-Garg (LG) inequality \cite{Leggett1985,Emary2014}.
The LG inequality is a temporal analog of the Bell inequality \cite{Bell1964,Clauser1969}, which concerns temporal correlations in a single system and serves as a test to distinguish quantum theory from ontological models. 
Although the original discussion by Leggett and Garg concerns operational correlations given by the sequential projective measurements, it is widely known that the same violation can be observed for algebraic correlations and related concepts \cite{Fritz2010,Hofmann2015,Halliwell2016,Pan2020,Ruskov2006,Williams2008,Palacios-Laloy2010,Goggin2011,Suzuki2012}.
In this work, we present the equivalence among such concepts, algebraic correlations, quasi-probabilities, weak value, and weak measurements from the perspective of the general framework of quantum/quasi-classical representations. 
Furthermore, we prove that the coincidence between the operational and algebraic correlations occurs only for dichotomic observables, thereby revealing the underlying structure of the equivalence among the ways of analyzing the quantum violation of the LG inequality.

This paper is organized as follows.
In section \ref{Invasiveness of a measurement}, we introduce the measure of invasiveness of a measurement and discuss its basic properties.
By using the measure of invasiveness, we present an upper bound on the difference between operational and algebraic correlations in section \ref{Relation between operational and algebraic quantum correlations}.
In section \ref{General framework of quantum/quasi-classical representation}, we introduce the notion of quasi-probabilities and review the general framework of quantum/quasi-classical representation.
We then present the relation between operational and algebraic (quasi-) probability distributions in section \ref{Relation between operational and algebraic (quasi-) probability distributions}.
In section \ref{A generalization to general measurements}, we present a generalization to general measurements.
In section \ref{Application to the Leggett-Garg inequality}, we discuss the quantum violation of the Leggett-Garg inequality.
Lastly, we give conclusion and discussion of the results in section \ref{Conclusion}.

Let \(\mathcal{H}\), \(\mathcal{B}(\mathcal{H}),\; \mathcal{B}_{\mathrm{sa}}(\mathcal{H})\), \(\mathcal{T}(\mathcal{H})\), \(\mathcal{K}(\mathcal{H}),\; \mathcal{K}_{\mathrm{sa}}(\mathcal{H})\), \(\mathcal{B}(\mathcal{T}(\mathcal{H}))\) and \(\mathcal{S}(\mathcal{H}):=\{\rho\in\mathcal{T}(\mathcal{H})\;|\;\rho\geq0,\;\mathrm{tr}[\rho]=1\}\) denote a Hilbert space, the space of bounded operators on \(\mathcal{H}\), the space of self-adjoint bounded operators on \(\mathcal{H}\), the space of trace class operators on \(\mathcal{H}\), the sapce of compact operators on \(\mathcal{H}\), the space of self-adjoint compact operators on \(\mathcal{H}\), the space of bounded operators on \(\mathcal{T}(\mathcal{H})\) and the state space of the quantum system identified by \(\mathcal{H}\), respectively.
We denote the operator norm and the trace norm by \(\norm{\cdot}\) and \(\norm{\cdot}_{1}\) , respectively.
We also denote the total variation of a signed measure \(\mu\) as \(\norm{\mu}_{\mathrm{TV}}\) and the spremum norm of a measurable function \(f\) by \(\norm{f}_{\infty}\).

A quantum measurement is represented as a positive instrument \(M=\{M_m\}_m\) and we denote the non-selective process of \(M\) as \(\Lambda_M:=\sum_m M_m\).
We denote the spectrum of a quantum observable \(A\in\mathcal{K}_{\mathrm{sa}}(\mathcal{H})\) by \(\sigma(A)\) and the spectral family associated with \(A\)  by \(P_A:=\{P_A(a)\}_{a\in \sigma(a)}\) and (L\"{u}der-) projective measurement of \(A\) as \(\Pi_A:=\{\Pi_A(a)\in\mathcal{B}(\mathcal{T}(\mathcal{H}))\}_{a\in \sigma(A)},\; \Pi_A(a):\rho\mapsto P_A(a)\rho P_a(a)\).
For simplicity, we write \(\Lambda_A:=\Lambda_{\Pi_A}\).

\section{Invasiveness of a measurement}
\label{Invasiveness of a measurement}
In this work, we focus on two-point correlations and restrict our attention to operational correlations arising from sequential projective measurements and simple forms of algebraic correlations.
\begin{dfn}[Operational quantum correlations]
We call \textit{operational quantum joint-probability} and \textit{operational quantum correlation of \(A\to B\)} as the joint probability and the correlation describes the statistics of sequential (L\"uders-) projective measurement of \(A\) and \(B\) and denote them as 
\begin{equation}
P_{\rho}^{A\to B}(a\to b):=\mathrm{tr}[\Pi_A(a)(\rho) P_B(b)]
\end{equation}
and
\begin{equation}
\langle A\to B\rangle_{\rho}^{\mathrm{op}}:=\sum_{a,b} abP_{\rho}^{A\to B}(a\to b),
\end{equation}
respectively.
\end{dfn}

\begin{dfn}[Algebraic quantum correlations]
For \(\alpha\in\mathbb{R}\), we denote
\begin{equation}
A\circ_{\alpha} B:=\alpha AB+(1-\alpha)BA.
\end{equation}
Then its expectation value
\(\langle  A\circ_{\alpha} B \rangle_{\rho}:=\mathrm{tr}\qty[\rho  A\circ_{\alpha} B]\)
is called \textit{algebraic quantum correlations} of quantum observables \(A\) and \(B\).
\end{dfn}

One of the dominant factors that causes difficulty in defining the correlation function is the inevitable \textit{invasiveness} of quantum measurements.
Although a classical measurement also alters the state as the increase of our knowledge, an ideal classical measurement is considered to be \textit{non-invasive} in the sense that the initial state remains unchanged under the non-selective measurement, that is, when the measurement is performed but its outcome is disregarded.
A rigorous formulation of invasiveness in this sense was discussed, for example,  in \cite{Kofler2013, Schild2015, Moreira2019}.

Here, we use the trace-norm distance between the initial state and the post-measurement state under the non-selective implementation as a measure of the invasiveness of a measurement:
\begin{dfn}[Invasiveness measure]
We define the \textit{invasiveness} of a measurement \(M:=\{M_m\}_m\) as
\begin{equation}
\mathrm{Inv}_{M}(\rho):=\norm{\Lambda_M(\rho)-\rho}_{1}
\end{equation}
\begin{equation}
\mathrm{Inv}(M)
:=\hspace{-2pt}\underset{\rho\in \mathcal{S}(\mathcal{H})}{\mathrm{sup}}\mathrm{Inv}_M(\rho)
=\norm{\Lambda_M-\mathrm{id}_{\mathcal{T}(\mathcal{H})}}.
\end{equation}
\end{dfn}

For simplicity, we write  \(\mathrm{Inv}_{\Pi_A}(\rho)\) as \(\mathrm{Inv}_A(\rho)\) and \(\mathrm{Inv}(\Pi_A)\) as \(\mathrm{Inv}(A)\).
We can regard \(M\) as invasive whenever \(\mathrm{Inv}(M)>0\), since this implies that the state changes after the measurement \(M\) even if we forget the outcome of the measurement.
This guarantees that the measurement \(M\) is not merely an epistemic update.
Moreover, this definition of invasiveness ensures that the \textit{no-signalling in time} condition \cite{Kofler2013} 
\begin{equation}
\sum_{m}P_{\rho}^{M\to N}(m\to n)=P_{\rho}^{N}(n)
\end{equation}
is satisfied for any subsequent measurement \(N:=\{N_n\}_n\) iff \(\mathrm{Inv}_M(\rho)=0\). 
As discussed in \cite{Schild2015, Moreira2019}, the invasiveness of a measurement is closely related to the coherence of a state.
For example,  \(\mathrm{Inv}_A(\rho)\) can also be interpreted as the measure of coherence \cite{Shao2015} of the state \(\rho\) in the \(A\)-basis.

We note that this definition of invasiveness is closely related to the notion of disturbance of an observable, which quantifies how an observable was disturbed by a measurement and is used frequently in the context of the uncertainty relation \cite{Ozawa2003}. 
The definition of the disturbance is given as
\begin{dfn}[disturbance operator \cite{Ozawa2003}]
We denote
\begin{equation}
\delta_M(A):=\Lambda_M^{\dagger}(A)-A
\end{equation}
and call it \textit{disturbance operator} of an observable \(A\) by the measurement \(M\).
\end{dfn}

\begin{dfn}[maximum disturbance]
We denote
\begin{equation}
\Delta_M(A;\rho):=\hspace{-2pt}\underset{f\in C_0(\sigma(A))}{\mathrm{sup}}\frac{\qty|\langle\delta_M(f(A))\rangle_\rho|}{\norm{f(A)}}
\end{equation}
and call it \textit{maximum disturbance} of \(A\)-observables in the state \(\rho\) by a measurement \(M\).
\end{dfn}
For simplicity, we write \(\delta_{\Pi_A}(B)\) as \(\delta_A(B)\) and \(\Delta_{\Pi_A}(B;\rho)\) as \(\Delta_A(B;\rho)\).
\(\Delta_M(A;\rho)\) quantifies how the projective measurement \(\Pi_A\) of an observable \(A\) is disturbed by the measurement \(M\).
As we will see later, \(\Delta_M(A;\rho)\) is particularly relevant when we focus on the probability distribution for the measurement of \(A\) since the assigned values themselves are not important in this case.

We can see that the invasiveness of a measurement and the disturbance of an observable by a measurement are related as follows:
\begin{prop}
\label{relation between invasiveness and disturbance}
\begin{equation}
\mathrm{Inv}_M(\rho)
=\hspace{-2pt}\underset{A\in\mathcal{K}(\mathcal{H})}{\mathrm{sup}}\frac{\qty|\langle\delta_M(A)\rangle_\rho|}{\norm{A}}
=\hspace{-2pt}\underset{A\in\mathcal{K}_{\mathrm{sa}}(\mathcal{H})}{\mathrm{sup}}{\Delta_M(A;\rho)}
\end{equation}
\end{prop}
\begin{proof}
From the isomorphism \(\mathcal{T}(\mathcal{H})\simeq\qty(\mathcal{K}(\mathcal{H}))^*\), we find
\begin{align*}
\mathrm{Inv}_M(\rho)
&=\underset{A\in\mathcal{K}(\mathcal{H})}{\mathrm{sup}}\frac{\mathrm{tr}\Bigl[\bigl(\Lambda_M(\rho)-\rho\bigr)A\Bigr]}{\norm{A}}\\
&=\underset{A\in\mathcal{K}(\mathcal{H})}{\mathrm{sup}}\frac{\qty|\langle\delta_M(A)\rangle_\rho|}{\norm{A}}.
\end{align*}
The second equality is obvious since \(C^*(A)\subset \mathcal{K}(\mathcal{H})\).
\end{proof}

\section{Relation between operational and algebraic quantum correlations}
\label{Relation between operational and algebraic quantum correlations}
In this section, we will see that the difference among several definitions of quantum correlations are bounded by the invasiveness of the measurement.
First, we give an algebraic description of the operational correlation:
\begin{prop}
\label{algebraic description of operational correlation}
\begin{equation}
\langle A\to B \rangle^{\mathrm{op}}_\rho=\langle  A\circ_{\alpha} B\rangle_{\Lambda_A(\rho)}.
\end{equation}
\end{prop}
\begin{proof}
By using
\begin{equation*}
\Pi_A(a)(\rho)=\Pi_A(\rho)P_A(a)=P_A(a)\Pi_A(\rho),
\end{equation*}
we find that
\begin{align*}
\langle A\to B \rangle^{\mathrm{op}}_\rho
&=\sum_{a,b}ab\,\mathrm{tr}\qty[\Pi_A(a)(\rho)P_B(b)]\\
&=\sum_{a,b}ab\,\mathrm{tr}\qty[\Pi_A(\rho)P_A(a)P_B(b)]\\
&=\langle  AB\rangle_{\Lambda_A(\rho)},\\
\langle A\to B \rangle^{\mathrm{op}}_\rho
&=\sum_{a,b}ab\,\mathrm{tr}\qty[P_A(a)\Pi_A(\rho)P_B(b)]\\
&=\langle  BA\rangle_{\Lambda_A(\rho)}.
\end{align*}
Thus, the statement holds for arbitrary \(\alpha\in\mathbb{R}\).
\end{proof}
From this relation, we also find that the maximum correlation that can be achievable by algebraic correlations is larger than that of the operational correlation as follows.

\begin{cor}
\label{inequality about maximum correlation}
\begin{equation}
\underset{\rho\in \mathcal{S}(\mathcal{H})}{\mathrm{sup}}|\langle A\to B\rangle^{\mathrm{op}}_\rho|
\leq \underset{\rho\in \mathcal{S}(\mathcal{H})}{\mathrm{sup}} |\langle A\circ_{\alpha} B\rangle_\rho|.
\end{equation}
\end{cor}
\begin{proof}
By Prop.~\ref{algebraic description of operational correlation}, we find
\begin{align*}
\underset{\rho\in \mathcal{S}(\mathcal{H})}{\mathrm{sup}}|\langle A\to B\rangle^{\mathrm{op}}_\rho|
&=\underset{\rho\in \mathrm{Ran}(\Lambda_A|_{\mathcal{S}(\mathcal{H})})}{\mathrm{sup}} |\langle A\circ_{\alpha} B\rangle_{\rho}|\\
&\leq \underset{\rho\in \mathcal{S}(\mathcal{H})}{\mathrm{sup}} |\langle A\circ_{\alpha} B\rangle_\rho|.
\end{align*}
\end{proof}
By using algebraic discription of the operational correlation, we can show that the difference between operational correlations and algebraic correlations are bounded by the product of the invasiveness of the measurement and the operator norm of \( A\circ_{\alpha} B\), which quantifies ``how large observables \(A\) and \(B\) can simultaneously be."
\begin{thm}
\label{upper bound of difference between operational and algebraic correlations}
\begin{equation}
\qty|\langle A\to B \rangle^{\mathrm{op}}_\rho-\langle  A\circ_{\alpha} B\rangle_\rho|
\leq \norm{ A\circ_{\alpha} B}\;\mathrm{Inv}_A(\rho).
\end{equation}
In particular, for \(\alpha\in[0,1]\),
\begin{equation}
\qty|\langle A\to B \rangle^{\mathrm{op}}_\rho-\langle  A\circ_{\alpha} B\rangle_\rho|
\leq \norm{A}\norm{B}\;\mathrm{Inv}_A(\rho).
\end{equation}
\end{thm}
\begin{proof}
By Prop.~\ref{algebraic description of operational correlation}, we find
\begin{align}
\qty|\langle A\to B \rangle^{\mathrm{op}}_\rho-\langle  A\circ_{\alpha} B\rangle_\rho|
&=\qty|\langle  A\circ_{\alpha} B\rangle_{\Lambda_A(\rho)}-\langle  A\circ_{\alpha} B\rangle_\rho|\notag\\
&=\qty|\mathrm{tr}\bigl[\qty(A\circ_{\alpha} B)\qty(\Lambda_A(\rho)-\rho)\bigr]|\notag\\
&\leq \norm{A\circ_{\alpha} B} \norm{\Lambda_A(\rho)-\rho}_{1}.\notag
\end{align}
The latter statement follows from \(\norm{A\circ_{\alpha}B}\leq|\alpha|\norm{AB}+|1-\alpha|\norm{BA}\leq\norm{A}\norm{B}\).
\end{proof}
The upper bound of the difference between two orders of operational correlations follows straightforwardly.
\begin{cor}
\label{upper bound of difference between two operational correlations}
\begin{equation}
|\langle A\to B \rangle^{\mathrm{op}}_\rho-\langle B\to A \rangle^{\mathrm{op}}_\rho|
\leq \norm{ A\circ_{\alpha} B}\Bigl(\mathrm{Inv}_A(\rho)+\;\mathrm{Inv}_B(\rho)\Bigr).
\end{equation}
\end{cor}
\begin{proof}
By Prop.~\ref{algebraic description of operational correlation}, we find
\begin{align}
&|\langle A\to B \rangle^{\mathrm{op}}_\rho-\langle B\to A \rangle^{\mathrm{op}}_\rho|\notag\\
=&\qty|\langle  A\circ_{\alpha} B\rangle_{\Lambda_A(\rho)}-\langle  A\circ_{\alpha} B\rangle_{\sum_{b}\Pi_B(b)(\rho)}|\notag\\
\leq &\norm{ A\circ_{\alpha} B}\Bigl(\norm{\Lambda_A(\rho)-\rho}_{1}+\norm{\Lambda_B(\rho)-\rho}_{1}\Bigr).\notag
\end{align}
\end{proof}

The above inequalities ensure that the difference among several definitions of correlation functions is due to the invasiveness of the measurement as expected, and they are quantitatively bounded by the measure of invasiveness.

\section{General framework of quantum/quasi-classical representation}
\label{General framework of quantum/quasi-classical representation}
Similarly to the case of the correlation functions, the genuine joint-probability distribution for the values of multiple quantum observables cannot be defined in general due to the incompatibility of quantum measurements.
\textit{Quasi-joint probability (QJP) distributions} \cite{Wigner1932,Kirkwood1933,Dirac1945,Margenau1961,Cohen1966,Lee2017,Lee2018} are quantum analogs of the joint-probability distribution for a tuple of incompatible observables.
While QJPs generally take complex values, they are similar to the genuine joint-probability distribution in the sense that all of their marginal distributions reproduce the standard Born rule for each observable.
Although QJPs have arbitrariness in their definitions due to the non-commutativity of quantum observables and have historically been defined in a heuristic manner, they are unified within the general framework of quantum/quasi-classical representations \cite{Lee2017,Lee2018}.

In contrast to the operational correlations which are defined through the actual measurement procedures, algebraic correlations do not have underlying joint-probability distributions.
However, within the framework of quantum/quasi-classical representations where the quasi-classicalization of the states and the quantization of multi-variable functions are dualistically understood \cite{Cohen1966,Lee2017,Lee2018}, we can regard QJPs as such underlying distributions that naturally reproduce algebraic correlations.

In this section, we briefly review the general framework of \textit{quantum/quasi-classical representations} \cite{Lee2017,Lee2018}.

In this framework, we first define the \textit{quasi-joint spectral distributions (QJSDs)} \cite{Lee2017,Lee2018} for a (generally non-commutative) tuple of quantum observables as an generalization of the spectral measures for (strongly) commutative cases as follows:
\begin{dfn}[Quasi-joint spectral distributions]
The quasi-joint spectral distributions for a tuple of quantum observables \(\bm{A}=(A_1,\cdots, A_n)\) is given by the inverse Fourier transform
\begin{equation}
\begin{aligned}
\#_{\bm{A}}(\bm{a})
&:=(\mathscr{F}^{-1}\hat{\#}_{\bm{A}})(\bm{a})\\
&=\frac{1}{\sqrt{2\pi}^n}\int_{\mathbb{R}^n} e^{i\bm{s}\cdot\bm{a}}\hat{\#}_{\bm{A}}(\bm{s})\;\mathrm{d}\bm{s}
\end{aligned}
\end{equation}
of the \textit{hashed operators}
\begin{equation}
\hat{\#}_{\bm{A}}(\bm{s}):=
\left\{
\begin{gathered}
\text{a suitable ``mixture" of the ``disintegrated"} \\
\text{components of the unitary groups }\\
e^{-is_1 A_1}, \cdots, e^{-i s_n A_n}
\end{gathered}
\right\}.
\end{equation}
\end{dfn}
Examples of hashed operators for the simplest case \(\bm{A}=(A,B)\) include
\begin{equation}
\hat{\#}_{(A,B)}(s,t)=
\left\{
\begin{aligned}
&e^{-isA}e^{-itB}, \\
&e^{-i\alpha sA}e^{-itB}e^{-i(1-\alpha)sA}, \\
&\frac{1+\alpha}{2}e^{-isA}e^{-it\hat{B}}+\frac{1-\alpha}{2}e^{-itB}e^{-isA}, \\
&e^{-i\overline{(sA+tB)}} = \lim_{N \to \infty}(e^{-i\frac{s}{N}A}e^{-i\frac{t}{N}B})^N, \\
&\frac{1}{2}\int_{-1}^1 e^{-i\frac{1-k}{2}sA} e^{-itB} e^{-i\frac{1+k}{2}sA} \,\mathrm{d}k.
\end{aligned}
\right.
\end{equation}
The first one generates the Kirkwood-Dirac distribution \cite{Kirkwood1933,Dirac1945}, the case of \(\alpha=0\) for the third one generates the Margenau-Hill distribution \cite{Margenau1961}, the fourth one corresponds to the Wigner-Weyl transform \cite{Wigner1932,Weyl1927} and the fifth one corresponds to the Born-Jordan quantization \cite{Born1925}.
Hereafter, we write first one and third one as \(\hat{\#}^{\mathrm{KD}}_{(A,B)}\) and \(\hat{\#}^{\alpha}_{(A,B)}\) and denote \(\#^{\alpha}_{(A,B)}:=\mathscr{F}^{-1}\hat{\#}^{\alpha}_{(A,B)}\) and \(\#^{\mathrm{KD}}_{(A,B)}:=\mathscr{F}^{-1}\hat{\#}^{\mathrm{KD}}_{(A,B)}\), respectively.

We note that the definitions of QJSDs involve arbitrariness due to the non-commutativity of quantum observables and in the case where every pair of observables strongly commutes, they reduce to the standard joint spectral measure.

By using QJSDs, QJPs and quantizations of multivariable functions are defined as an generalization of the Born rule and the functional calculus, respectively, as follows \cite{Lee2017,Lee2018}.
\begin{dfn}[Quasi-joint probability distribution]
Given a QJSD \(\#_{\bm{A}}(\bm{a})\), the quasi-joint probability distribution for a tuple of observables \(\bm{A}=(A_1,\cdots, A_n)\) in a state \(\rho\) is given as
\begin{equation}
P_{\rho}^{\#_{\bm{A}}}(\bm{a}):=\mathrm{tr}[\rho\#_{\bm{A}}(\bm{a})].
\end{equation}
\end{dfn}
\begin{dfn}[Quantization]
Given a QJSD \(\#_{\bm{A}}(\bm{a})\), the quantization of a multivariable function \(f(\bm{a})\) of a tuple of observables \(\bm{A}=(A_1,\cdots, A_n)\) is given as
\begin{equation}
f(\bm{A})^{\#_{\bm{A}}}:=\int_{\mathbb{R}^n} f(\bm{a})\#_{\bm{A}}(\bm{a})\;\mathrm{d}\bm{a}.
\end{equation}
\end{dfn}
For simplicity, we write \(P_{\rho}^{\mathrm{KD}}:=P_{\rho}^{\#^{\mathrm{KD}}_{(A,B)}}\), \(P_{\rho}^{\alpha}:=P_{\rho}^{\#^{\alpha}_{(A,B)}}\), \(f(\bm{A})^{\mathrm{KD}}:=f(\bm{A})^{\#_{\bm{A}}^{\mathrm{KD}}}\), and \(f(\bm{A})^{\alpha}:=f(\bm{A})^{\#_{\bm{A}}^{\alpha}}\).

Given a ``basis" set of observables \(\bm{A}\), the quasi-classicalization of quantum states \(\rho\mapsto P_{\rho}^{\#_{\bm{A}}}(\bm{a})\) and the quantization of classical observables \(f(\bm{a})\mapsto f(\bm{A})^{\#_{\bm{A}}}\) are dualistically understood:
It is equivalent whether to compute ``quasi-expectation value" of an observable \(f(\bm{a})\) in a state \(\rho\) by quasi-classicalizing the state or to compute it by quantizing the observable
\begin{equation}
\int_{\mathbb{R}^n} f(\bm{a})P_{\rho}^{\#_{\bm{A}}}(\bm{a})\;\mathrm{d}\bm{a}
=\mathrm{tr}[\rho f(\bm{A})^{\#_{\bm{A}}}]
=:\langle f(\bm{a}) \rangle_{\rho}^{\#_{\bm{A}}},
\end{equation}
for a fixed generator \(\#_{\bm{A}}\) of the representation (see Fig.\ref{duality of quasi-classicalization and quantization}).
\;\\

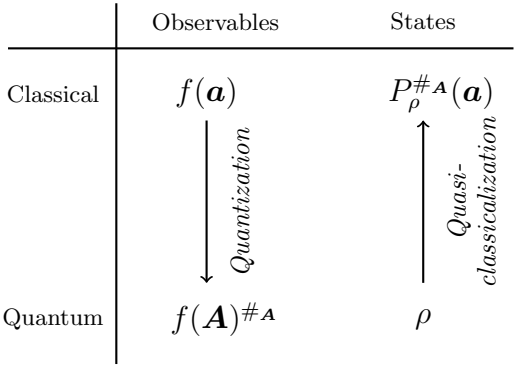
\begin{figure}[H]
\centering
\begin{tikzpicture}[scale=1.2]
\draw [-,thick] (-3.2,2.5) -- (2.4,2.5);
\draw [-,thick] (-2,3) -- (-2,-1);
\draw [->,thick] (-1,1.7) -- (-1,-0.1);
\draw [->,thick] (1.4,-0.1) -- (1.4,1.7);
\draw (-0.9,2.8) node{Observables};
\draw (1.4,2.8) node{States};
\draw (-1,2) node{\large{\(f(\bm{a})\)}};
\draw (-0.8,-0.5) node{\large{\(f(\bm{A})^{\#_{\bm{A}}}\)}};
\draw (1.6,2) node{\large{\(P_{\rho}^{\#_{\bm{A}}}(\bm{a})\)}};
\draw (1.4,-0.5) node{\large{\(\rho\)}};
\draw (-2.7,2) node{Classical};
\draw (-2.7,-0.5) node{Quantum};
\draw (-0.6,0.8) node[rotate=90]{\normalsize{\textit{Quantization}}};
\draw (1.8,0.8) node[rotate=90]{\normalsize{\textit{Quasi-}}};
\draw (2.1,0.8) node[rotate=90]{\normalsize{\textit{classicalization}}};
\end{tikzpicture}
\caption{The duality of quasi-classicalization and quantization as an adjoint operations.}
\label{duality of quasi-classicalization and quantization}
\end{figure}

Since \(\mathrm{tr}[\rho f(\bm{A})^{\#_{\bm{A}}}]\) in the case of \(f_0(\bm{a})=\Pi_{i=1}^{n} a_i\) are exactly what  we call algebraic correlations, we can regard \(P_{\rho}^{\#_{\bm{A}}}(\bm{a})\) as the underlying distribution that generates the algebraic correlation.
Indeed, we can show that the quantization \(f_0(a,b)^{\#}\) of the two-point correlations \(f_0(a,b)=ab\) is restricted to simple forms \( A\circ_{\alpha} B\) as following (mathematically not rigorous) discussions.
First, we observe
\begin{align}
f_0(a,b)^{\#}
&:=\int ab \#(a,b) \,\mathrm{d}a\mathrm{d}b\notag\\
&=\frac{1}{2\pi} \int ab \qty(\int e^{isa}e^{itb}\hat{\#}(s,t) \,\mathrm{d}s\mathrm{d}t)\,\mathrm{d}a\mathrm{d}b\notag\\
&=-\frac{1}{2\pi} \int \int e^{isa}e^{itb}\pdv{s}\pdv{t}\hat{\#}(s,t) \,\mathrm{d}s\mathrm{d}t\,\mathrm{d}a\mathrm{d}b\notag\\
&=-\int \delta(s)\delta(t)\pdv{s}\pdv{t}\hat{\#}(s,t) \,\mathrm{d}s\mathrm{d}t\notag\\
&=-\left. \qty(\pdv{s}\pdv{t}\hat{\#}(s,t))\right|_{s=t=0}.
\end{align}
Since
\begin{equation}
\begin{aligned}
&\left. \qty(\pdv{s}\pdv{t}e^{-i\alpha_1sA}e^{-i\beta_1tB}\cdots e^{-i\alpha_nsA}e^{-i\beta_ntB}) \right|_{s=t=0}\\
&=-\qty(\sum_{i\leq j}\alpha_i\beta_j AB+\sum_{i>j}\alpha_i\beta_j BA),
\end{aligned}
\end{equation}
\(f_0(a,b)^{\#}\) can be written as \( A\circ_{\alpha} B\) for some \(\alpha\in \mathbb{R}\), at least for the representations generated by the hashed operators that can be written as a well-converging limit of the affine combination of the above simple forms.
Similar discussion holds for \(f_0(\bm{a})=\Pi_{i=1}^{n} a_i\).
We note that, since \((AB)^\dagger=BA\), \(\mathrm{Re}\langle  A\circ_{\alpha} B\rangle_{\rho}\) do not depend on \(\alpha\in\mathbb{R}\), and thus the real part of the two-point algebraic correlation is the same for all representations.
However, this does not hold for more than three point correlations.
For example, a three-point algebraic correlation defined by different orderings of quantization \(\langle\sigma_x\sigma_x\sigma_z\rangle_{\ket{z+}}=\langle \sigma_z\rangle_{\ket{z+}}=1\) and \(\langle\sigma_x\sigma_z\sigma_x\rangle_{\ket{z+}}=\langle -\sigma_z\rangle_{\ket{z+}}=-1\) disagrees.

We also introduce the notion of \textit{quasi-conditional expectations} \cite{Lee2017,Lee2018} as a quantum extension of the standard conditional expectation:
\begin{equation}
\begin{aligned}
&\mathbb{E}_{\rho}^{\#_{\bm{A}}}[A_i\,|\,(a_1,\cdots,a_{i-1},a_{i+1},\cdots,a_n)]\\
&\;\;\;\;\;\;\;\;\;\;:=\frac{\int_{\mathbb{R}} a_i P_{\rho}^{\#_{\bm{A}}}(a_1,\cdots,a_n)\; \mathrm{d}{a_i}}{\int_{\mathbb{R}} P_{\rho}^{\#_{\bm{A}}}(a_1,\cdots,a_n)\; \mathrm{d}{a_i}}.
\end{aligned}
\end{equation}
This includes Aharonov's \textit{weak value} \cite{Aharonov1988,Aharonov1990,Dressel2014} as a special case.
It can be understood as {quasi-conditional expectations} of the Kirkwood-Dirac representation \cite{Steinberg1995,Hosoya2011,Morita2013}:
\begin{equation}
\begin{aligned}
{A^{\mathrm{w}}}_{\rho}(b)
:=\mathbb{E}_{\rho}^{\mathrm{KD}}[A\,|\,b]
&=\frac{\int_{\mathbb{R}} a P_{\rho}^{\mathrm{KD}}(b,a)\;\mathrm{d}a}{P_{\rho}^B(b)}\\
&=\frac{\bra{b}A\ket{\psi}}{\bra{b}\ket{\psi}}.
\end{aligned}
\end{equation}
Here, the last equality holds for the case of non-degenerated eigenvalue \(b\) of the observable \(B\) and pure states \(\rho=\ket{\psi}\bra{\psi}\).
Aharonov's weak value was introduced as a ``value of an observable between the initial and final states" and is known to be measurable by \textit{weak measurement} \cite{Aharonov1988,Ritchie1991} using post-selection, which is known to be applicable to quantum metrology \cite{Hosten2008,Dixon2009,Lundeen2011,Lee2014,Mori2019,Xu2020,Arvidsson-Shukur2020}.

We remark that, in general, the support of the distribution and the spectrum of the observables in consideration are in non-trivial relation \cite{Hofmann2014,Umekawa2024}, presenting a form of non-classicality of the quantum system.
Quasi-joint probability distributions can take non-zero values outside the spectrum of observables, the values that the observables can take, \textit{i.e.}, there exists a value \(b\) for which \(P_{\rho}(b)=0\) but \(P_{\rho}^{\#_{(A,B)}}(a,b)\neq 0\) for some \(a\) and \(\rho\), other than Kirkwood-Dirac type representations \(\#^{\alpha}_{(A,B)}\).
This leads to the quasi-conditional expectation not being well-defined at \(b\),

\section{Relation between operational and algebraic (quasi-) probability distributions}
\label{Relation between operational and algebraic (quasi-) probability distributions}
In this section, we present both upper and lower bounds for the difference among underlying operational and algebraic (quasi-)probability distributions that reproduce correlation functions.

First, we give an algebraic description of the operational probability distribution:
\begin{prop}
\label{algebraic description of operational probability}
\begin{equation}
P_\rho^{A\to B}(a\to b)=P_{\Lambda_A(\rho)}^{\alpha}(a,b).
\end{equation}
\end{prop}
\begin{proof}
We find that
\begin{align*}
P_\rho^{A\to B}(a\to b)
&=\mathrm{tr}[\Pi_A(a)(\rho)P_B(b)]\\
&=\mathrm{tr}\qty[\Pi_A(\rho)P_A(a)P_B(b)]\\
&=P_{\Lambda_A(\rho)}^{KD}(a,b).
\end{align*}
Since \(P_\rho^{A\to B}(a\to b)\in\mathbb{R}\), it follows that \(\mathrm{Im}[P_{\Lambda_A(\rho)}^{KD}(a,b)]=0\).
Thus, we obtain 
\begin{align*}
P_{\Lambda_A(\rho)}^{\alpha}(a,b)&=\mathrm{Re}[P_{\Lambda_A(\rho)}^{KD}(a,b)]+\alpha\,\mathrm{Im}[P_{\Lambda_A(\rho)}^{KD}(a,b)]\\
&=P_\rho^{A\to B}(a\to b)
\end{align*}
for arbitrarily \(\alpha\in\mathbb{R}\).
\end{proof}

By using algebraic discription of the operational probability, we show that the difference among operational and algebraic (quasi-)probabilities are bounded by the invasiveness and non-commutativity of quantum observables. 
\begin{thm}
\label{upper bound of difference between operational and algebraic probabilities}
For \(\alpha\in[0,1]\),
\begin{equation}
\begin{aligned}
&|P_{\rho}^{A\to B}(a\to b)-P_{\rho}^{\alpha}(a,b)|\\
&\leq \mathrm{min}
\left\{
\begin{aligned}
&\norm{[P_A(a),\;P_B(b)]},\\
&\norm{P_A(a)\circ_{\alpha}P_B(b)}
\end{aligned}
\right\}
\mathrm{Inv}_A(\rho).
\end{aligned}
\end{equation}
\end{thm}
\begin{proof}
By Prop.\ref{algebraic description of operational probability}, we find
\begin{align*}
&|P_{\rho}^{A\to B}(a\to b)-P_{\rho}^{\alpha}(a,b)|\\
&\;\;\;\;\;\;=\qty|P_{\Lambda_A(\rho)}^{\alpha}(a,b)-P_{\rho}^{\alpha}(a,b)|\\
&\;\;\;\;\;\;\leq \norm{\Lambda_A(\rho)-\rho}_{1}\norm{P_A(a)\circ_{\alpha}P_B(b)}.
\end{align*}
Since \(P_A(a)(\Lambda_A(\rho)-\rho)P_A(a)=0\), we also find
\begin{align}
&|P_{\rho}^{A\to B}(a\to b)-P_{\rho}^{\alpha}(a,b)\notag|\\
&\;\;\;\;\;\leq\Bigl|P_{\Lambda_A(\rho)}^{\mathrm{KD}}(a,b)-P_{\rho}^{\mathrm{KD}}(a,b)\Bigr|\notag\\
&\;\;\;\;\;=\qty|\mathrm{tr}\qty[\Bigl(\Lambda_A(\rho)-\rho\Bigr)P_A(a)P_A(a)P_B(b)]|\notag\\
\begin{split}
&\;\;\;\;\;=\biggl|\mathrm{tr}\qty[\Bigl(\Lambda_A(\rho)-\rho\Bigr)P_A(a)[P_A(a),\,P_B(b)]]\\
&\;\;\;\;\;\quad+\mathrm{tr}\qty[P_A(a)\Bigl(\Lambda_A(\rho)-\rho\Bigr)P_A(a)P_B(b)]\biggr|\\
\end{split}\notag\\
&\;\;\;\;\;\leq \norm{\Lambda_A(\rho)-\rho}_{1}\norm{P_A(a)[P_A(a),\,P_B(b)]}\notag\\
&\;\;\;\;\;\leq \norm{\Lambda_A(\rho)-\rho}_{1}\norm{[P_A(a),\,P_B(b)]}.\notag
\end{align}
\end{proof}
The upper bound of differences between two orders of operational probabilities follows straightforwardly.
\begin{cor}
\label{upper bound of difference between two operational probabilities}
For \(\alpha\in[0,1]\),
\begin{equation}
\begin{aligned}
&|P_{\rho}^{A\to B}(a\to b)-P_{\rho}^{\alpha}(a,b)|\\
&\leq \mathrm{min}
\left\{
\begin{aligned}
&\norm{[P_A(a),\;P_B(b)]},\\
&\norm{P_A(a)\circ_{\alpha}P_B(b)}
\end{aligned}
\right\}
\Bigl(\mathrm{Inv}_A(\rho)+\mathrm{Inv}_B(\rho)\Bigr).
\end{aligned}
\end{equation}
\end{cor}
\begin{proof}
The statement follows from Thm.\ref{upper bound of difference between operational and algebraic probabilities} and the triangle inequality.
\end{proof}
The upper bounds of the differences among (quasi-)joint probability distributions derived here are given by the product of two factors. 
While the second factor evaluates how the first measurement disturbs the initial state, the first factor that involves a minimization term represents the relation between two observables under consideration.
While the commutator term \(\norm{[P_A(a),\;P_B(b)]}\) can be interpreted as evaluating how incompatible they are, the product term \(\norm{P_A(a)\circ_{\alpha}P_B(b)}\) can be roughly interpreted as evaluating how large ``the probability for the value of \(A\) to be \(a\) and \(B\) to be \(b\) at the same time" is.

We also provide a lower bound of the difference among operational and algebraic (quasi-)probability distributions, which can be interpreted as a kind of uncertainty relation.
\begin{prop}
\label{lower bound of difference between operational and algebraic probabilities}
\begin{equation}
\norm{P_{\rho}^{A\to B}-P_{\rho}^{\#}}_{\mathrm{TV}}\geq \Delta_A(B;\rho).
\end{equation}
\end{prop}
\begin{proof}
We find that
\begin{align}
&\norm{P_{\rho}^{A\to B}-P_{\rho}^{\#}}_{\mathrm{TV}}\notag\\
&\;\;\;\;\;\geq \frac{\qty|\int f(a,b)\;\mathrm{d}\qty(P_{\rho}^{A\to B}(a\to b)-P_{\rho}^{\#}(a,b))|}{\norm{f}_{\infty}}\notag\\
&\;\;\;\;\;= \frac{\qty|\int f(a,b)\;\mathrm{d}\qty(P_{\Lambda_A(\rho)}^{KD}(a,b)-P_{\rho}^{\#}(a,b))|}{\norm{f}_{\infty}}\notag\\
&\;\;\;\;\;=\frac{\qty|\langle f^{\mathrm{KD}}(A,B) \rangle_{\Lambda_A(\rho)}-\langle f^{\#}(A,B) \rangle_{\rho}|}{\norm{f}_{\infty}}\notag
\end{align}
for any measurable function \(f\).
By restricting \(f\) to \(f\in C_0(\sigma(B))\), we obtain
\begin{align}
\norm{P_{\rho}^{A\to B}-P_{\rho}^{\#}}_{\mathrm{TV}}
&\geq \frac{\qty|\langle f(B) \rangle_{\Lambda_A(\rho)}-\langle f(B) \rangle_{\rho}|}{\norm{f}_{\infty}}\notag\\
&\geq\frac{\qty|\mathrm{tr}\qty[f(B)\Bigl(\Lambda_A(\rho)-\rho\Bigr)]|}{\norm{f(B)}}\notag\\
&=\frac{\qty|\langle\delta_A(f(B))\rangle_\rho|}{\;\;\norm{f(B)}}.\notag
\end{align}
\end{proof}
The lower bound of the difference between two orders of operational probabilities follows straightforwardly.
\begin{thm}
\label{lower bound of difference between two operational probabilities}
\begin{equation}
\norm{P_{\rho}^{A\to B}-P_{\rho}^{B\to A}}_{\mathrm{TV}}
\geq\mathrm{max}\qty{\Delta_A(B;\rho),\,\Delta_B(A;\rho)}.
\end{equation}
\end{thm}
\begin{proof}
We find that
\begin{align}
&\norm{P_{\rho}^{A\to B}-P_{\rho}^{B\to A}}_{\mathrm{TV}}\notag\\
&\;\;\;\;\;\geq \frac{\qty|\int f(a,b)\;\mathrm{d}\qty(P_{\rho}^{A\to B}(a\to b)-P_{\rho}^{B\to A}(b\to a))|}{\norm{f}_{\infty}}\notag\\
&\;\;\;\;\;=\frac{\qty|\langle f^{KD}(A,B) \rangle_{\Lambda_A(\rho)}-\langle f^{KD}(A,B) \rangle_{\Lambda_B(\rho)}|}{\norm{f}_{\infty}}.\notag
\end{align}
By taking \(f(a,b)=f(a)\), we obtain from \(\Lambda_A^{\dagger}(f(A))=f(A)\) that 
\begin{align}
\norm{P_{\rho}^{A\to B}-P_{\rho}^{B\to A}}_{\mathrm{TV}}
&\geq \frac{\qty|\langle f(A) \rangle_{\Lambda_A(\rho)}-\langle f(A) \rangle_{\Lambda_B(\rho)}|}{\norm{f(A)}}\notag\\
&=\frac{\qty|\mathrm{tr}\qty[\qty(f(A)-\Lambda_B^{\dagger}(f(A)))\rho]|}{\norm{f(A)}}\notag\\
&=\frac{\qty|\langle\delta_B(f(A))\rangle_\rho|}{\;\;\norm{f(A)}}.\notag
\end{align}
In the same way, we also find
\begin{equation*}
\norm{P_{\rho}^{A\to B}-P_{\rho}^{B\to A}}_{\mathrm{TV}}
\geq\frac{\qty|\langle\delta_A(f(B))\rangle_\rho|}{\;\;\norm{f(B)}}.
\end{equation*}
\end{proof}
The above results show that, in the case where an observable is disturbed in expectation by the measurement of another observable, definitions of their joint probability distributions must diverge.
We note that the lower bound given here is very tight.
Indeed, it always becomes equal in the case of the qubit system.
Detailed calculations for the example of the qubit system are presented in the Appendix \ref{example of qubit system}.

We remark that the coincidence among the definitions of their joint probability distributions does not imply state dependent commutativity given by neither \(\langle [A,B]\rangle_{\rho}=0\) nor \(\norm{[A,B]}_{\rho}=0\).
An explicit counterexample is presented in the Appendix \ref{example of qubit system}.

\section{A generalization to general measurements}
\label{A generalization to general measurements}
Here, we present a generalization of the above results for obsevables to the generic measurements.
We consider two quantum measurements \(M=\{M_m\}_m\) and \(N=\{N_n\}_n\) and investigate the difference among their operational and algebraic (quasi-)joint probability distributions.
Here, we use
\begin{equation}
P_{\rho}^{\alpha}(m, n):=\mathrm{tr}\Bigl[\rho \Bigl(\alpha E_M(m)E_N(n)+(1-\alpha)E_{N}(n)E_M(m)\Bigr)\Bigr]
\end{equation}
as algebraic quasi-joint probability distributions for \(M\) and \(N\), where \(E_M=\{E_M(m)\}_m\) and \(E_N=\{E_N(n)\}_n\) denote the POVMs corresponding to \(M\) and \(N\), respectively.

We show that the difference among operational and algebraic (quasi-)joint probability distributions for generic quantum measurements are bounded by a term that can be interpreted as evaluating non-repeatability of the measurement in addition to the same invasiveness term appearing in the previous inequalities.
\begin{prop}
\begin{equation}
\begin{aligned}
&|P_{\rho}^{M\to N}(m,n)-P_{\rho}^{\alpha}(m,n)|\\
&\;\;\;\;\leq \norm{E_M(m)\circ_{\alpha}E_N(n)}\,\mathrm{Inv}_A(\rho)+\norm{E_N(n)}\\
&\;\;\;\;\;\;\;\;\cdot\norm{M_m(\rho)\circ_{\alpha}\hspace{-7pt}\sum_{m': m'\neq m}\hspace{-7pt}E_M(m')+\hspace{-7pt}\sum_{m': m'\neq m}\hspace{-7pt}M_{m'}(\rho)\circ_{\alpha}E_M(m)}_{1}
\end{aligned}
\end{equation}
\end{prop}
\begin{proof}
We find
\begin{equation}
\begin{aligned}
&|P_{\Lambda_M(\rho)}^{\alpha}(m,n)-P_{\rho}^{\alpha}(m,n)|\notag\\
&\;\;\;\;\leq \norm{E_M(m)\circ_{\alpha}E_N(n)}_{\mathcal{B}(\mathcal{H})}\norm{\Lambda_M(\rho)-\rho}_{1}.\notag
\end{aligned}
\end{equation}
We also find
\begin{align}
&|P_{\rho}^{M\to N}(m,n)-P_{\Lambda_M(\rho)}^{\alpha}(m,n)|\notag\\
&\;\;\;\;=\mathrm{tr}\qty[\qty(M_m(\rho)-\sum_{m'}M_{m'}(\rho)\circ_{\alpha}E_{M}(m))E_N(n)]\notag\\
\begin{split}
&\;\;\;\;\leq \norm{M_m(\rho)\circ_{\alpha}\hspace{-7pt}\sum_{m': m'\neq m}\hspace{-7pt}E_M(m')+\hspace{-7pt}\sum_{m': m'\neq m}\hspace{-7pt}M_{m'}(\rho)\circ_{\alpha}E_M(m)}_{1}\\
&\;\;\;\;\;\;\;\;\cdot \norm{E_N(n)}.
\end{split}\notag
\end{align}
Thus, we obtain the desired inequality from the triangle inequality.
\end{proof}
The first term of our bound is the same as in the case of observables, \textit{i.e.}, when both \(M\) and \(N\) are projective measurements, which evaluates the invasiveness of the prior measurement.
Since \(M_m(\rho)\circ_{\alpha}E_M(m')\) can be regarded as the contribution associated with obtaining outcome \(m'\) in the post-measurement state conditioned on outcome \(m\), the second term of our bound can be viewed as evaluating the non-repeatability of \(M\).
However, as its precise operational meaning remains unclear, further research on the generalization to general measurements is required.

\section{Application to the Leggett-Garg inequality}
\label{Application to the Leggett-Garg inequality}
In this section, we discuss the quantum violation of the Leggett-Garg (LG) inequality \cite{Leggett1985,Emary2014}.
The LG inequality is an inequality regarding temporal correlations in a single system and holds under certain assumptions regarding realism.
We consider a sequence of times \(t_1<t_2<t_3\) and \(\pm 1\)-valued observables \(A_1,\, A_2,\, A_3\) at each time and denote the ``temporal correlation" between \(A_i\) and \(A_j\) at time \(t_i\) and \(t_j\), respectively, as \(C_{ij}\).
Then the simplest version of the LG inequality is given by \cite{Leggett1985}
\begin{equation}
\label{Leggett-Garg inequality}
K:=C_{12}+C_{23}-C_{13}\leq 1.
\end{equation}
Leggett and Garg argued that the above inequality holds under the assumptions of (i) macrorealism per se  and (ii) non-invasive measurability and (iii) induction, which together ensure the existence of a standard joint-probability distribution for the values of \(A_1,\, A_2,\, A_3\) at \(t_1,\, t_2,\, t_3\), respectively.

The problem here lies in the definition of the ``temporal correlation" as previously discussed.
If we measure at all of the times \(t_1,\, t_2,\, t_3\) in every run of the experiment and define ``temporal correlation" \(C_{ij}\) from the statistics of these measurement outcomes, then LG inequality \eqref{Leggett-Garg inequality} holds even in quantum theory since their standard joint-probability distribution for their outcomes exists.
Leggett and Garg argued \cite{Leggett1985} that if we measure two of the times among \(t_1,\, t_2,\, t_3\) in each run of the experiment and define ``temporal correlation" \(C_{ij}\) from the statistics of measurement outcomes for the runs of experiments where the measurement is conducted at \(t_i\) and \(t_j\), then quantum theory violates the LG inequality and \(K\) can take values as large as \(3/2\). 

Although the original discussion by Leggett and Garg described above concerns the operational correlation given by the sequential projective measurements, there have been many works that deal with algebraic correlations and related concepts \cite{Fritz2010,Hofmann2015,Halliwell2016,Pan2020,Ruskov2006,Williams2008,Palacios-Laloy2010,Goggin2011,Suzuki2012}.
Some works used weak measurement \cite{Ruskov2006,Williams2008,Pan2020,Palacios-Laloy2010,Goggin2011,Suzuki2012} in order to guarantee the non-invasiveness criteria and to insist that the quantum violation of the LG inequality implies the violation of the assumption of realism itself.
These works show that the same violation of the LG inequality can be observed for weak measurement as in the standard sequential projective measurement case, and it is also experimentally verified \cite{Palacios-Laloy2010,Goggin2011,Suzuki2012}.
Some of these works discussed the relation between anomalous weak values and the LG violation \cite{Pan2020,Williams2008,Goggin2011}, both of which signify the non-classicality of quantum systems.
Since the violation of the LG inequality is due to the non-existence of standard joint-probability distributions for the value \(A_1,\, A_2,\, A_3\) at \(t_1,\, t_2,\, t_3\), respectively, there are also some works that analyzed the quantum violation of the LG inequality by using quasi-probabilities \cite{Hofmann2015,Halliwell2016,Pan2020,Suzuki2012} as a quantum extension of standard joint-probability distributions.

Here, we first present that the analysis using algebraic correlation, quasi-probabilities, weak value, and the weak measurement is essentially equivalent.
From the general framework of quantum/quasi-classical representations \cite{Lee2017,Lee2018}, we observe that
\begin{equation}
\begin{aligned}
\label{equivalence between algebraic correlation, quasi-probability, and quasi-conditional expectation}
C_{ij}^{\#}
&:=\Bigl\langle f_0(a_i,a_j)^{\#}\Bigr\rangle_{\rho}\\
&=\sum_{a_i, a_j}a_ia_jP^{\#}_{\rho}(a_i,a_j)\\
&=\sum_{a_j}P(a_j)a_j\mathbb{E}_{\rho}^{\#}[A_i|a_j],
\end{aligned}
\end{equation}
where \(f_0(a_i,a_j)=a_ia_j\) denotes the classical two-point correlation function.
Thus, we find that it is equivalent whether we use algebraic correlation defined as the expectation of the quantization of the product of observables, or the quasi-joint probability distribution for observables in consideration, or corresonding quasi-expectation value once we fix the representation \(\#\).

Moreover, as discussed in Sec.~\ref{General framework of quantum/quasi-classical representation}, since quantization of two-point correlation function is restricted to the forms of \(f_0(a_i,a_j)^{\#}=A_i(t_i)\circ_{\alpha} A_j(t_j)\), the real part of \(C_{ij}^{\#}\) do not depend on the representation \(\#\).
This result shows that it is equivalent among representations in analyzing the quantum violation of the LG inequality, and that it ensures that the same violation can be observed whichever quasi-joint probability distributions we use.

We note that, however, as discussed in Sec.~\ref{General framework of quantum/quasi-classical representation}, other than Kirkwood-Dirac type representation \(\#^{\alpha}\), there is a value \(a_j\) in which \(\mathbb{E}[A_i|a_j]^{\#}\) is not well-defined.
Thus in that case, the last line of \eqref{equivalence between algebraic correlation, quasi-probability, and quasi-conditional expectation} is not well-defined, and if we sum over only for \(a_j\in \sigma(A_i(t_i))\), the same violation cannot be observed (see Appendix \ref{Quantum violation of the Leggett-Garg inequality from quasi-conditional expectations} for an explicit demonstration).

We now discuss the equivalence among operational correlations and algebraic correlations in the analysis of the quantum violation of the LG inequality.
In \cite{Fritz2010,Halliwell2016}, it is pointed out that for \(\pm 1\)-valued observables \(A\) and \(B\), their algebraic correlation coincide with the operational correlations
\begin{equation}
\biggl\langle \frac{\{A,B\}}{2}\biggr\rangle_{\rho}=\langle A\to B\rangle^{\mathrm{op}}_{\rho}.
\end{equation}
This fact certifies the equivalence among operational correlations and algebraic correlations when we analyze the quantum violation of the LG inequality.
Here, we prove the converse: such an equality holds only for \(\pm \alpha\)-valued observables for some \(\alpha\).
\begin{thm}
For a nontrivial quantum observable \(A\in \mathcal{K}_{\mathrm{sa}}(\mathcal{H})\backslash \{\lambda I\,|\,\lambda\in\mathbb{R}\}\), conditions (i)--(iii) below are equivalent to each other:
\begin{enumerate}
\item[(i)] Operational correlation and algebraic correlation always coincide, i.e.
\begin{equation}
^\forall B\in\mathcal{B}(\mathcal{H}),\; ^\forall \rho\in \mathcal{S}(\mathcal{H}),\;\langle A\to B\rangle_\rho^{\mathrm{op}}=\left\langle \frac{\{A,B\}}{2}\right\rangle_\rho
\end{equation}
\item[(ii)] The operational correlation and algebraic correlation between \(A\) and every \(\pm1\)-valued observable \(B\) coincide, i.e.
\begin{equation}
\begin{aligned}
&\sigma(B)=\{1,-1\}\\
&\implies ^\forall \rho\in \mathcal{S}(\mathcal{H}),\; \langle A\to B\rangle_\rho^{\mathrm{op}}=\left\langle \frac{\{A,B\}}{2}\right\rangle_\rho
\end{aligned}
\end{equation}
\item[(iii)] \(A\) is \(\pm \alpha\)-valued, i.e.
\begin{equation}
^\exists \alpha>0\; s.t.\;\sigma(A)=\{\alpha,-\alpha\}
\end{equation}
\end{enumerate}
\end{thm}
\begin{proof}
\;\\
\vspace{-\baselineskip}
\begin{description}
\item[((i)\(\implies\)(ii))]
\;\par
Obvious.
\item[((ii)\(\implies\)(iii))]
\;\par
First, we prove that the sum of any pair of eigenvalues of \(A\) must be \(0\).
From Prop.~\ref{algebraic description of operational correlation}, it follows that
\begin{align}
&\qty|\langle A\to B\rangle_\rho^{\mathrm{op}}-\left\langle \frac{\{A,B\}}{2}\right\rangle_\rho|\notag\\
&=\qty|\mathrm{tr}\qty[\frac{\{A,B\}}{2}\qty(\Lambda_A(\rho)-\rho)]|\notag\\
&=\qty|\mathrm{tr}\qty[\frac{B}{2}\qty{A,\,\qty(\Lambda_A(\rho)-\rho)}]|\notag\\
&=\qty|\mathrm{tr}\qty[\frac{B}{2}\sum_{a'}a'\Bigl\{P_A(a'),\,\Bigl(P_A(a')\rho P_A(a')-\rho\Bigr)\Bigr\}]|\notag\\
\begin{split}
&=\Biggl|\mathrm{tr}\Biggl[\frac{B}{2}\sum_{a'}a'\Bigl(P_A(a')\rho \bigl(I-P_A(a')\bigr)\\
&\;\;\;\;\;\;\;\;\;\;\;\;\;\;\;\;\;\;\;\;+\bigl(I-P_A(a')\bigr)\rho P_A(a')\Bigr)\Biggr]\Biggr|.
\end{split}\notag
\end{align}
We take two eigenvectors \(\ket{\alpha}\in \mathrm{Ran}(P_A(\alpha))\) and \(\ket{\alpha'}\in\mathrm{Ran}(P_A(\alpha'))\) for \(\alpha\neq \alpha'\) of \(A\) (since \(A\in\mathcal{K}_{\mathrm{sa}}(\mathcal{H})\backslash\{\lambda I|\lambda\in\mathbb{R}\}\), we can assume \(A\) has at least two eigenvalues) and take \(B\) and \(\rho\) as
\begin{align*}
B&:=\ket{\alpha}\bra{\alpha'}+\ket{\alpha'}\bra{\alpha},\\
\rho&:=\qty(\frac{\ket{\alpha}+\ket{\alpha'}}{\sqrt{2}})\qty(\frac{\bra{\alpha}+\bra{\alpha'}}{\sqrt{2}}).
\end{align*}
Then for this \(\rho\) and \(B\), we find
\begin{align}
&\qty|\langle A\to B\rangle_\rho^{\mathrm{op}}-\left\langle \frac{\{A,B\}}{2}\right\rangle_\rho|\notag\\
\begin{split}
&=\Biggl|\mathrm{tr}\Biggl[\frac{B}{2}\Biggl(\alpha\Bigl(P_A(\alpha)\rho \bigl(I-P_A(\alpha)\bigr)+\bigl(I-P_A(\alpha)\bigr)\rho P_A(\alpha)\Bigr)\\
&\;\;\;\;\;+\alpha'\Bigl( P_A(\alpha')\rho \bigl(I-P_A(\alpha')\bigr)+\bigl(I-P_A(\alpha')\bigr)\rho P_A(\alpha')\Bigr)\Biggr)\Biggr]\Biggr|\\
\end{split}\notag\\
&=\qty|\frac{\alpha+\alpha'}{2}|.\notag
\end{align}
Therefore, we find that \(\alpha+\alpha'=0\) holds if condition (ii) holds.
Assume \(A\) have more than three different eigenvalues \(\alpha, \alpha', \alpha''\).
From the above argument, we obtain \(\alpha+\alpha'=\alpha+\alpha''=0\), which contradicts to \(\alpha'\neq\alpha''\).
Thus, we conclude that \(A\) have only two eigenvalues \(\alpha,\alpha'\) that satisfies \(\alpha+\alpha'=0\).
Since a self-adjoint compact operator that has only a finite number of eigenvalues does not contain \(0\) in its continuous spectrum, we conclude that condition (iii) holds.
\item[((iii)\(\implies\)(i))]
\;\par
From the calculation above and an equality \(P_A(\alpha)+P_A(-\alpha)=I\), we find
\begin{align}
&\qty|\langle A\to B\rangle_\rho^{\mathrm{op}}-\left\langle \frac{\{A,B\}}{2}\right\rangle_\rho|\notag\\
\begin{split}
&=\Biggl|\mathrm{tr}\Biggl[\frac{B}{2}\Biggl(\alpha\Bigl(P_A(\alpha)\rho P_A(-\alpha)+P_A(-\alpha)\rho P_A(\alpha)\Bigr)\\
&\;\;\;\;\;\;\;\;\;\;-\alpha\Bigl( P_A(-\alpha)\rho P_A(\alpha)+P_A(\alpha)\rho P_A(-\alpha)\Bigr)\Biggr)\Biggr]\Biggr|\\
\end{split}\notag\\
&=0\notag.
\end{align}
\end{description}
\end{proof}
This theorem establishes a novel equivalence condition for the coincidence between operational and algebraic correlations.
Moreover, as an application, this result reveal that the equivalence among the ways of observing quantum violation of the LG inequality is because the LG inequality concerns only two-point correlation of  bivalent observables.

\section{Conclusion}
\label{Conclusion}
In this work, we have investigated the relations among several definitions of correlation functions and their underlying (quasi-)joint probability distributions, based on the general framework of quantum/quasi-classical representations.
We show that the differences among them are bounded by the invasiveness of measurements, a quantification of how the system is disturbed beyond the mere increase of our knowledge.
We also prove that the coincidence between operational and algebraic correlations is strictly limited to that of dichotomic observables.
As a prominent application, we reveal the structure of the equivalence among different ways of observing the quantum violation of the Leggett-Garg (LG) inequality.

Our results establish an operational justification of widely-used algebraic concepts in a quantitative manner, and bridge between different approaches for the foundations of quantum theory.
Thus, as we have done for the LG inequality, various applications to a wide range of topics in quantum foundations can be expected.
It is also desirable to extend our discussion to generic measurements beyond projective measurements of observables and to other scenarios beyond sequential measurements.

\begin{acknowledgments}
This work was supported by FoPM, WINGS Program, the University of Tokyo (S.~U.) and JSPS Grant-in-Aidfor Scientific Research (KAKENHI), Grant No. JP22K13970 (J.~L.).
\end{acknowledgments}

\appendix

\section{Example: qubit system}
\label{example of qubit system}
Here, we consider the qubit system as an example.
Without loss of generality, we can assume \(A=\alpha_0 I+\alpha_1\sigma_z\) and \(B=\beta_0 I +\beta_1 (\cos\theta\sigma_z+\sin\theta\sigma_x)\) for certain \(\alpha_0,\, \,\beta_0\in\mathbb{R},\,\alpha_1,\,\beta_1>0,\,\theta\in[0,2\pi)\) by taking a suitable basis.
We remark that \(P_A(\alpha_0\pm\alpha_1)=P_{z\pm}:=P_{\sigma_z}(\pm1)\) and \(P_B(\beta_0\pm\beta_1)=P_{\theta\pm}:=P_{\sigma_\theta}(\pm1)\) (\(\sigma_\theta:=\cos\theta\sigma_z+\sin\theta\sigma_x\)).

The difference between operational joint-probability distributions that represent the statistics of sequential projective measurements is given by
\begin{equation}
\begin{aligned}
P_\rho^{A\to B}(\alpha_0\pm\alpha_1\to\beta_0+\beta_1)
&=P_\rho^{\sigma_z\to \sigma_\theta}(\pm1\to1)\\
&=\mathrm{tr}\qty[\rho\frac{I\pm\sigma_z}{2}\frac{I+\sigma_\theta}{2}\frac{I\pm\sigma_z}{2}]\\
&=\frac{1\pm\cos\theta}{4}(1\pm\langle \sigma_z\rangle_\rho),\\
P_\rho^{A\to B}(\alpha_0\pm\alpha_1\to\beta_0-\beta_1)
&=\frac{1\mp\cos\theta}{4}(1\pm\langle \sigma_z\rangle_\rho),\\
P_\rho^{B\to A}(\beta_0+\beta_1\to\alpha_0\pm\alpha_1)
&=\frac{1\pm\cos\theta}{4}(1+\langle \sigma_\theta\rangle_\rho),\\
P_\rho^{B\to A}(\beta_0-\beta_1\to\alpha_0\pm\alpha_1)
&=\frac{1\mp\cos\theta}{4}(1-\langle \sigma_\theta\rangle_\rho),
\end{aligned}
\end{equation}
\begin{align}
&\norm{P_\rho^{A\to B}-P_\rho^{B\to A}}_{\mathrm{TV}}\notag\\
\begin{split}
&=\frac{1+\cos\theta}{2}\qty|(1-\cos\theta)\langle \sigma_z\rangle_\rho-\sin\theta\langle \sigma_x\rangle_\rho|\\
&\;\;\;\;\;+\frac{1-\cos\theta}{2}\qty|(1+\cos\theta)\langle \sigma_z\rangle_\rho+\sin\theta\langle \sigma_x\rangle_\rho|
\end{split}\notag\\
&=\mathrm{max}\qty{\Bigl|\sin\theta\langle\sigma_x\rangle_\rho\Bigr|,\,\qty|\sin\theta\Bigl(-\sin\theta\langle\sigma_{z}\rangle_\rho+\cos\theta\langle\sigma_{x}\rangle_\rho\Bigr)|
}.
\end{align}

The upper bound of the difference between them derived from our inequality is calculated as follows.
First, the operator norms of the commutator and the anti-commutator of the spectral family of the observables in consideration are given by
\begin{equation}
\begin{aligned}
\frac{\{P_{z\pm},P_{\theta+}\}}{2}
&=\frac{1}{4}\qty((1\pm\cos\theta)(I\pm\sigma_z)+\sin\theta\sigma_x),\\
\frac{\{P_{z\pm},P_{\theta-}\}}{2}
&=\frac{1}{4}\qty((1\mp\cos\theta)(I\pm\sigma_z)-\sin\theta\sigma_x),
\end{aligned}
\end{equation}
\begin{equation}
\begin{aligned}
P_{z\pm}[P_{z\pm},P_{\theta+}]
&=\frac{I\pm\sigma_z}{2}\qty(\pm i\sin\theta\frac{\sigma_y}{2})
=\frac{\sin\theta}{4}\qty(\sigma_x\pm i\sigma_y),\\
P_{z\pm}[P_{z\pm},P_{\theta-}]
&=\frac{\sin\theta}{4}\qty(-\sigma_x\mp i\sigma_y),\\
P_{\theta+}[P_{z\pm},P_{\theta+}]
&=\frac{\sin\theta}{4}\qty(\pm\cos\theta\sigma_x\mp\sin\theta\sigma_z\pm i\sigma_y),\\
P_{\theta-}[P_{z\pm},P_{\theta-}]
&=\frac{\sin\theta}{4}\qty(\pm\cos\theta\sigma_x\mp\sin\theta\sigma_z\mp i\sigma_y),
\end{aligned}
\end{equation}
\begin{align}
&\sum_{a,b}\norm{\frac{\{P_A(a),P_B(b)\}}{2}}\notag\\
&=1+\frac{1}{2}\sqrt{(1+\cos\theta)^2+\sin^2\theta})+\frac{1}{2}\sqrt{(1-\cos\theta)^2+\sin^2\theta}),
\end{align}
\begin{equation}
\begin{aligned}
\sum_{a,b}\norm{P_A(a)[P_A(a),P_B(b)]}
&=\sqrt{2}\qty|\sin\theta|,\\
\sum_{a,b}\norm{P_B(b)[P_B(b),P_A(a)]}
&=\sqrt{2}\qty|\sin\theta|.
\end{aligned}
\end{equation}
The invasiveness of the projective measurement of the observables in consideration is given by 
\begin{equation}
\begin{aligned}
&\Lambda_A(\rho)
=\frac{1}{2}(I+\langle\sigma_z\rangle_\rho \sigma_z),\\
&\Lambda_B(\rho)
=\frac{1}{2}\Bigl(I+\qty(\cos\theta\langle\sigma_z\rangle_\rho+\sin\theta\langle\sigma_x\rangle_\rho) \qty(\cos\theta\sigma_z+\sin\theta\sigma_x)\Bigr),
\end{aligned}
\end{equation}
\begin{equation}
\begin{aligned}
\norm{\Lambda_A(\rho)-\rho}_{1}
&=\sqrt{\langle\sigma_x\rangle_\rho^2+\langle\sigma_y\rangle_\rho^2},\\
\norm{\Lambda_B(\rho)-\rho}_{1}
&=\sqrt{\qty(-\sin\theta\langle\sigma_z\rangle_\rho+\cos\theta\langle\sigma_x\rangle_\rho)^2+\langle\sigma_y\rangle_\rho^2}.
\end{aligned}
\end{equation}

The lower bound of our inequality becomes
\begin{align}
\Delta_A(B;\rho)
&=\underset{f\in C_0(\sigma(B))}{\mathrm{sup}}\frac{\qty|\langle\delta_A(f(B))\rangle_\rho|}{\norm{f(B)}}\notag\\
&=\underset{\beta_0,\beta_1\in\mathbb{C}}{\mathrm{sup}}\frac{\qty|\langle\delta_A(\beta_0I+\beta_1\sigma_\theta)\rangle_\rho|}{\norm{\beta_0I+\beta_1\sigma_\theta}}\notag\\
&=\underset{\beta_0,\beta_1\in\mathbb{C}}{\mathrm{sup}}\frac{\qty|\beta_1\sin\theta\langle\sigma_x\rangle_\rho|}{\mathrm{max}\{\qty|\beta_0+\beta_1|, \qty|\beta_0-\beta_1|\}}\notag\\
&=\qty|\sin\theta\langle\sigma_x\rangle_\rho|,\\
\Delta_B(A;\rho)
&=\underset{\alpha_0,\alpha_1\in\mathbb{C}}{\mathrm{sup}}\frac{\qty|\langle\delta_B(\alpha_0I+\alpha_1\sigma_z)\rangle_\rho|}{\norm{\alpha_0I+\alpha_1\sigma_z}}\notag\\
&=\qty|\sin\theta\Bigl(-\sin\theta\langle\sigma_{z}\rangle_\rho+\cos\theta\langle\sigma_{x}\rangle_\rho\Bigr)|.
\end{align}

Thus, our inequality
\begin{equation}
\begin{aligned}
&\;\;\;\;\;\mathrm{max}\qty{\Delta_A(B;\rho),\,\Delta_B(A;\rho)}\\
&\leq \norm{P_\rho^{A\to B}-P_\rho^{B\to A}}_{\mathrm{TV}}\\
&\leq \mathrm{min}
\left\{
\begin{aligned}
&\sum_{a,b}\norm{\frac{\{P_A(a),P_B(b)\}}{2}},\\
&\sum_{a,b}\norm{P_A(a)[P_A(a),P_B(b)]}
\end{aligned}
\right\}
\mathrm{Inv}_A(\rho)\\
&\;\;\;\;\;+\mathrm{min}
\left\{
\begin{aligned}
&\sum_{a,b}\norm{\frac{\{P_A(a),P_B(b)\}}{2}},\\
&\sum_{a,b}\norm{P_B(a)[P_A(a),P_B(b)]}
\end{aligned}
\right\}
\mathrm{Inv}_B(\rho)
\end{aligned}
\end{equation}
becomes
\begin{equation}
\begin{aligned}
&\;\;\;\;\;\mathrm{max}\qty{\Bigl|\sin\theta\langle\sigma_x\rangle_\rho\Bigr|,\,\qty|\sin\theta\Bigl(-\sin\theta\langle\sigma_{z}\rangle_\rho+\cos\theta\langle\sigma_{x}\rangle_\rho\Bigr)|
}\\
&\leq\mathrm{max}\qty{\Bigl|\sin\theta\langle\sigma_x\rangle_\rho\Bigr|,\,\qty|\sin\theta\Bigl(-\sin\theta\langle\sigma_{z}\rangle_\rho+\cos\theta\langle\sigma_{x}\rangle_\rho\Bigr)|
}\\
&\leq \mathrm{min}
\left\{
\begin{aligned}
&\sqrt{2}\qty|\sin\theta|,\\
&1+\frac{1}{2}\sqrt{(1+\cos\theta)^2+\sin^2\theta})\\
&\;\;\;\;\;+\frac{1}{2}\sqrt{(1-\cos\theta)^2+\sin^2\theta})
\end{aligned}
\right\}
\biggl(\sqrt{\langle\sigma_x\rangle_\rho^2+\langle\sigma_y\rangle_\rho^2}\\
&\;\;\;\;\;+\sqrt{\qty(-\sin\theta\langle\sigma_z\rangle_\rho+\cos\theta\langle\sigma_x\rangle_\rho)^2+\langle\sigma_y\rangle_\rho^2}\biggr).
\end{aligned}
\end{equation}
We can see that the lower bound of this inequality is indeed an equality in this case (See, Fig.\ref{difference among joint probabilities}).
\begin{figure}
\centering
\includegraphics[width=60mm]{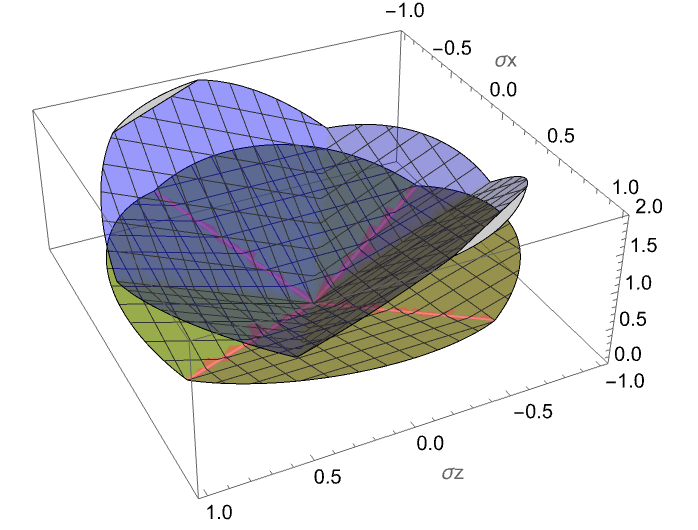}
\caption{The application of our inequality to the qubit system for the case of \(\theta=\frac{\pi}{3}\) in the states satisfies \(\langle \sigma_y \rangle_\rho =0\). The red surface describes the difference between the two operational probabilities. The blue and green ones shows the upper and the lower bounds derived from our inequality, respectively. Here, red and green ones coincide everywhere.}
\label{difference among joint probabilities}
\end{figure}

The state-dependent non-commutativity of observables in consideration are given as
\begin{equation}
\begin{aligned}
\langle [A,B]\rangle_{\rho}
&=i2\alpha_1\beta_1\sin\theta\langle \sigma_z \rangle_\rho,\\
\norm{[A,B]}_{\rho}
&:=\mathrm{tr}[-([A,B])^2\rho]=4\alpha_1^2\beta_1^2\sin^2\theta.
\end{aligned}
\end{equation}
Thus, we find that, even though two orders of operational joint-probability distributions coincide \(\norm{P_\rho^{A\to B}-P_\rho^{B\to A}}_{\mathrm{TV}}=0\), state-dependent non-commutativity do not vanish, \textit{i.e.}, \(\langle [A,B]\rangle_{\rho}\neq0\) and \(\norm{[A,B]}_{\rho}\neq0\), for the state \(\rho=\ket{y+}\bra{y+}\) and \(\theta\neq 0,\pi\).

\section{Quantum violation of the Leggett-Garg inequality from quasi-conditional expectations}
\label{Quantum violation of the Leggett-Garg inequality from quasi-conditional expectations}
Here, we demonstrate the analysis of the quantum violation of the Leggett-Garg (LG) inequality based on quasi-conditional expectations.
In each run of the experiments, we measure at \(t_1\) and \(t_3\), and assume the value at \(t_2\) as its quasi-coditional expectation.
This matches the original philosophy of Aharonov's weak value as ``time-symmetric formulation" of quantum theory \cite{Aharonov1988,Aharonov1990}.
In the case of the Kirkwood-Dirac representation, the quasi-conditional expectation \(\mathbb{E}_{\ket{a_1}}^{\mathrm{KD}}[A_2(t_2)\,|\,a_3]\) can be measured by the weak measurement of \(A_2\) at \(t_2\) by taking pre- and post-selected states as the post-measurement states \(\ket{a_1}\) and \(\ket{a_3}\) of projective measurements of \(A_1\) and \(A_3\) at \(t_1\) and \(t_3\), respectively.

In this scenario, \(K\) can be written as
\begin{equation}
\begin{aligned}
K^{\#}&:=\sum_{a_1,a_3}\Bigl(a_1\mathbb{E}_{\ket{a_1}}^{\#}[A_2(t_2)\,|\,a_3]+\mathbb{E}_{\ket{a_1}}^{\#}[A_2(t_2)\,|\,a_3]a_3-a_1a_3\Bigr)\\
&\;\;\;\;\;\;\;\;\;\;\;\;\;\;\;\cdot P_{\rho}^{A_1(t_1)\to A_3(t_3)}(a_1\to a_3).
\end{aligned}
\end{equation}

Here, as a simple example, we consider a qubit system governed by a Hamiltonian \(H=\omega \sigma_x\), and correlations among a dichotomic observable \(\sigma_z\) at each time, \(t_1=0, t_2=t, t_3=T\).
For simplicity, we assume that we obtain the outcome \(+1\) for the first measurement at \(t_1=0\).
Then, \(K\) can be written as
\begin{align}
\label{LG from quasi-conditional expectation}
K^{\#}
&=P_{\ket{z+}}^{\sigma_z(T)}(+1)\Bigl(1\cdot E_{\ket{z+}}^{\#}\bigl[\sigma_z(t)\,|+1\bigr]\notag\\
&\;\;\;\;\;+E_{\ket{z+}}^{\#}\bigl[\sigma_z(t)\,|+1\bigr]\cdot 1+1\cdot 1\Bigr)+P_{\ket{z+}}^{\sigma_z(T)}(-1)\notag\\
&\;\;\;\;\;\cdot \Bigl(1\cdot E_{\ket{z+}}^{\#}\bigl[\sigma_z(t)\,|+1\bigr]+E_{\ket{z+}}^{\#}\bigl[\sigma_z(t)\,|+1\bigr]\cdot (-1)\notag\\
&\;\;\;\;\;+1\cdot (-1)\Bigr)\notag\\
\begin{split}
&=|\bra{z+}U(T)\ket{z+}|^2\Bigl(2E_{\ket{z+}}^{\#}\bigl[\sigma_z(t)\,|+1\bigr]-1\Bigr)\\
&\;\;\;\;\;+|\bra{z-}U(T)\ket{z+}|^2.
\end{split}
\end{align}
Here, \(U(T)=e^{-iHT}\) denotes the time evolution unitary.
For the case where the Hamiltonian is given by \(H=\omega \sigma_x\), \(K\) becomes
\begin{align}
K^{\#}
&=\qty(\cos \frac{\omega T}{2})^2\Bigl(2E_{\ket{z+}}^{\#}\bigl[\sigma_z(t)\,|+1\bigr]-1\Bigr)+\qty(\sin \frac{\omega T}{2})^2.
\end{align}
Here, quasi-conditional expectations of the Kirkwood-Dirac representation \(\#^{\mathrm{KD}}\) and the semi-symmetrized representation \(\hat{\#}_{(A,B)}^{\mathrm{SS}}(s,t)=e^{-i\frac{t}{2}B}e^{-isA}e^{-i\frac{t}{2}B}\) are given as
\begin{align}
E_{\ket{z+}}^{\mathrm{KD}}\bigl[\sigma_z(t)\,|+1\bigr]
&=\frac{\bra{z+}U(T-t)\sigma_zU(t)\ket{z+}}{\bra{z+}U(T)\ket{z+}}\notag\\
&=\frac{\cos\qty(\omega \qty(\frac{T}{2}- t))}{\cos\frac{\omega T}{2}},\\
\begin{split}
E_{\ket{z+}}^{\mathrm{SS}}\bigl[\sigma_z(t)\,|+1\bigr]
&=\frac{1}{4\qty(\cos \frac{\omega T}{2})^2}\Bigl(\cos \omega t+2\cos(\omega (T-t))\\
&\;\;\;\;\;+\cos(\omega (2T-t))\Bigr).
\end{split}
\end{align}
Thus, \(K\) becomes
\begin{align}
K^{\mathrm{KD}}
&=2\cos \qty(\frac{\omega T}{2}) \cos \qty(\omega\qty(\frac{T}{2}-t))-\cos \omega T,\\
\begin{split}
K^{\mathrm{SS}}
&=\frac{1}{2}\cos \omega t+\cos(\omega (T-t))+\frac{1}{2}\cos(\omega (2T-t))\\
&\;\;\;\;\;-\cos \omega T.
\end{split}
\end{align}
While \(K\) takes the maximum value \(K^{\mathrm{KD}}_{\mathrm{max}}=\frac{3}{2}\) at \(t=\frac{T}{2}=\frac{\pi}{3}\) and violates the LG inequality for the Kirkwood-Dirac representation, the maximum value of \(K\) for the semi-symmetrized representation is \(K^{\mathrm{SS}}_{\mathrm{max}}=1\) at \(t=T\) and does not violate the LG inequality.

The difference between \(K^{\mathrm{KD}}\) and \(K^{\mathrm{SS}}\) is due to the relationship between the support of the distribution and the spectrum of the observables \cite{Umekawa2024}.
While the Kirkwood-Dirac distribution \(P^{\mathrm{KD}}\) is supported on \(\sigma(\sigma_z(t))\times \sigma(\sigma_z(T))=\{\pm1\}\times\{\pm1\}\), the semi-symmetrized distribution \(P^{\mathrm{SS}}\) takes non-zero values also at \((\pm1,0)\).
Thus, if we take the sum only over \(\sigma(\sigma_z(t))\times \sigma(\sigma_z(T))\) in the summation of \eqref{LG from quasi-conditional expectation}, we cannot reproduce the same violation.
This result suggests that, in the semi-symmetrized representation, the non-classical features of the quantum system are encoded in the support of the quasi-joint probability distribution outside the spectrum of the observables.


\bibliography{Ope-Alg_Cor}

\end{document}